\documentclass[11pt]{article}%
\usepackage{amsmath}
\usepackage{amssymb}
\usepackage{amsfonts}

\usepackage{cite}
\usepackage{graphicx}
\usepackage{float}
\usepackage{longtable}
\usepackage{hyperref}

\usepackage{algorithmicx}
\usepackage[ruled,vlined]{algorithm2e}
\SetAlgoCaptionSeparator{ }

\usepackage{fullpage}

\newcommand{\tref}[1]{Table~\ref{#1}}
\newcommand{\sref}[1]{Section~\ref{#1}}

\newenvironment{algo}[1][!htbp]
  {
   \begin{algorithm}[#1]%
  }{\end{algorithm}}

\providecommand{\U}[1]{\protect\rule{.1in}{.1in}}
\newtheorem{theorem}{Theorem}

\newtheorem{definition}{Definition}

\newtheorem{lemma}{Lemma}

\newtheorem{problem}{Problem}

\newenvironment{proof}[1][Proof]{\textbf{#1.} }{\ \rule{0.5em}{0.5em}}

\begin{document}
\title{Quantum State Optimization and Computational Pathway Evaluation for Gate-Model Quantum Computers}
\author{Laszlo Gyongyosi\thanks{School of Electronics and Computer Science, University of Southampton, Southampton SO17 1BJ, U.K., and Department of Networked Systems and Services, Budapest University of Technology and Economics, 1117 Budapest, Hungary, and MTA-BME Information Systems Research Group, Hungarian Academy of Sciences, 1051 Budapest, Hungary.}
}
\date{}

\maketitle
\begin{abstract}
A computational problem fed into a gate-model quantum computer identifies an objective function with a particular computational pathway (objective function connectivity). The solution of the computational problem involves identifying a target objective function value that is the subject to be reached. A bottleneck in a gate-model quantum computer is the requirement of several rounds of quantum state preparations, high-cost run sequences, and multiple rounds of measurements to determine a target (optimal) state of the quantum computer that achieves the target objective function value. Here, we define a method for optimal quantum state determination and computational path evaluation for gate-model quantum computers. We prove a state determination method that finds a target system state for a quantum computer at a given target objective function value. The computational pathway evaluation procedure sets the connectivity of the objective function in the target system state on a fixed hardware architecture of the quantum computer. The proposed solution evolves the target system state without requiring the preparation of intermediate states between the initial and target states of the quantum computer. Our method avoids high-cost system state preparations and expensive running procedures and measurement apparatuses in gate-model quantum computers. The results are convenient for gate-model quantum computations and the near-term quantum devices of the quantum Internet.
\end{abstract}

\section{Introduction}
\label{sec1}
Quantum computers \cite{qcomputer,a1,a2,a4,refibm,qcadd1,qcadd2,qcadd3,ref9,ref10} utilize the fundamentals of quantum mechanics to perform computations \cite{ref1,ref2,ref3,ref4,ref5,ref6,ref7,ref8,refa1}. For experimental gate-model quantum computer architectures and the near-term quantum devices of the quantum Internet \cite{puj1,puj2,refn7,ref21,qn1,qn2,qn3,qn4,refp5,refp6,refp7,nadd62,nadd62b,kris1,kris2,ca1,ca2,ca3,ca4,ca5,refn5,refn3,sat,telep,refn1,refn2,refn4,refn6,refqirg,van1,van2,van12,net2,net3,net4,net5,net6,net7,net8,net9,net10}, gate-based architectures provide an implementable solution to realize quantum computations \cite{ref9,ref9b,ref10,ref11,ref12,ref13,ref14,ref15,ref16,ref17,ref18,ref19,ref20,ref21,ref22, a1,a2,a3,a4,a5,a6,a7,a8,a9,a10,refp1,refp2,refp4,ibm2,petz,shor}. In a gate-model quantum computer the operations are realized via a sequence of quantum gates, and each quantum gate represents a unitary transformation \cite{ref10,ref11,ref12,ref13,ref14,ref15,ref16,ref17,ref18,ref19,ref20,ref21,ref22,ref23,ref24,uj1,uj2,uj3,uj4}. The input of a quantum computer is a quantum system realized via several quantum states, and the unitaries of the quantum computer change the initial system state into a specific state \cite{ref9,ref10,ref11,ref12}. The output quantum system is then measured by a measurement array. 

A computational problem fed into a quantum computer defines an objective function with a particular connectivity (computational pathway) \cite{ref10}. The solution of this computational problem in the quantum computer involves identifying an objective function with a target value that is subject to be reached. To achieve the target objective function value, the quantum computer must reach a particular system state such that the gate parameters of the unitary operations satisfy the target value. These optimal gate parameter values of the unitary operations of the quantum computer identify the optimal state of the quantum computer. This optimal system state is referred to as the target system state of the quantum computer. Finding the target system state involves multiple measurement rounds and iterations, with high-cost system state preparations\footnote{Note, the term "quantum state preparation" in the current context refers to a quantum state determination method. It is because the aim of the proposed procedure is the determination of an optimal state of the quantum computer, i.e., the optimal values of the gate-parameters of the unitaries of the quantum computer, see also \cite{ref10}.}, quantum computations, and measurement procedures. Therefore, optimizing the determination procedure of the target system state is essential for gate-model quantum computers.

Here, we define a method for state determination and computational path evaluation for gate-model quantum computers. The aim of state determination is to find a target system state for a quantum computer such that the pre-determined target objective function value is reached. The aim of the computational path evaluation is to find the connectivity of the objective function in the target system state on the fixed hardware architecture \cite{ref10} of the quantum computer. To resolve these issues, we define a framework that utilizes the theory of kernel methods \cite{ref25,ref26,ref27,ref28,ref29,ref30,ref31,ref32,ref33,ref34,ref35} and high-dimensional Hilbert spaces. In traditional theoretical computer science, kernel methods represent a useful and low computational-cost tool in statistical learning, signal processing theory and machine learning. We prove that these methods can also be utilized in gate-model quantum computations for particular problems.

The novel contributions of our manuscript are as follows:
\begin{enumerate}
\item \textit{We define a method for optimal quantum state determination and computational path evaluation for near-term quantum computers.} 

\item \textit{The proposed state determination method finds a target system state for a quantum computer at a given target objective function value.}

\item \textit{The computational pathway evaluation finds the connectivity of the objective function in the target system state on the fixed hardware architecture of the quantum computer.} 

\item \textit{The proposed solution evolves the target system state of the quantum computer without requiring the preparation of intermediate system states between the initial and target states of the quantum computer.}

\item \textit{The method avoids high-cost system state preparations, expensive running procedures and measurement rounds in gate-model quantum computers.}

\item \textit{The results are useful for gate-model quantum computers and the near-term quantum devices of the quantum Internet}.
\end{enumerate}

This paper is organized as follows. In \sref{relw}, related works are summarized. \sref{sec2} presents the problem statement. \sref{sec3} discusses the results. Finally, \sref{sec4} concludes the paper. Supplemental information is included in the Appendix.

\section{Related Works} 
\label{relw}
The related works are summarized as follows. 

\subsection{Gate-Model Quantum Computers}
The model of gate-model quantum computer architectures and the construction of algorithms for qubit architectures are studied in \cite{ref10}. The proposed system model of the work also serves as a reference for our system model. Some related preliminaries can also be found in \cite{ref11,ref12}. 

In \cite{ref9}, the authors defined a gate-model quantum neural network. The proposed system model is a quantum neural network realized via a gate-model quantum computer. 

In \cite{ref9b}, the authors studied a gate-model quantum algorithm called the “Quantum Approximate Optimization Algorithm” (QAOA) and its connection with the Sherrington-Kirkpatrick (SK) \cite{sk} model. The results serve as a framework for analyzing the QAOA, and can be used for evaluating the performance of QAOA on more general problems.

The behavior of the objective function value of the QAOA algorithm for some specific cases has been studied in \cite{a5}. As the authors concluded, for some fixed parameters and instances drawn from a particular distribution, the objective function value is concentrated such that typical instances have almost the same value of the objective function. 

Further performance analyses of the QAOA algorithm can be found in \cite{a7,a8}. 
Practical implementations connected to gate-model quantum computing and the QAOA algorithm can be found in \cite{a9,a10}.

In \cite{nadd7}, the authors studied methods quantum computing based hybrid solution methods for large-scale discrete-continuous optimization problems. The results are straightforwardly applicable for gate-model quantum computers. As the authors concluded, the proposed quantum computing methods have high computational efficiency in terms of solution quality and computation time, by utilizing the unique features of both classical and quantum computers.

A recent experimental quantum computer implementation has been demonstrated in \cite{qcomputer}. The results of the work confirmed the quantum supremacy \cite{a1,a2} of quantum computers over traditional computers in particular problems. 

The work of \cite{a4} gives a summary on quantum computing technologies in the NISQ (Noisy Intermediate-Scale Quantum) era and beyond.  

\subsection{Quantum State Preparation}

In \cite{nadd1}, the authors studied the utilization of reinforcement learning in different phases of quantum control. The authors studied the performance of reinforcement learning in the problem of finding short, high-fidelity driving protocol from an initial to a target state in non-integrable many-body quantum systems of interacting qubits. As the authors concluded, the performance of the proposed reinforcement learning method is comparable to optimal control methods.

In \cite{nadd2}, the authors studied the question of efficient variational simulation of non-trivial quantum states. The results represent an efficient and general route for preparing non-trivial quantum states that are not adiabatically connected to unentangled product states. The system model integrates a feedback loop between a quantum simulator and a classical computer. As the authors concluded, the proposed results are experimentally realizable on near-term quantum devices of synthetic quantum systems.

In \cite{nadd3}, the problem of simulated quantum computation of molecular energies is studied. While, on a traditional computer the calculation time for the energy of atoms and molecules scales exponentially with system size, on a quantum computer it scales polynomially. The authors demonstrated that such chemical problems can be solved via quantum algorithms using modest numbers of qubits.

In \cite{nadd4}, the authors studied the modeling and feedback control design for quantum state preparation. The work describes the modeling methods of controlled quantum systems under continuous observation, and studies the design of feedback controls that prepare particular quantum states. In the proposed analysis, the field-theoretic model is subjected to statistical inference and is ultimately controlled.

For an information theoretical analysis of quantum optimal control, see \cite{nadd5}. In this work, the authors studied quantum optimal control problems and the solving methods. The authors showed that if an efficient classical representation of the dynamics exists, then optimal control problems on many-body quantum systems can be solved efficiently with finite precision. As the authors concluded, the size of the space of parameters necessary to solve quantum optimal control problems defined on pure, mixed states and unitaries is polynomially bounded from the size of the of the set of reachable states in polynomial time.

In \cite{nadd6}, the authors studied the complexity of controlling quantum many-body dynamics. As the authors found, arbitrary time evolutions of many-body quantum systems can be reversed even in cases when only part of the Hamiltonian can be controlled. The authors also determined a lower bound on the control complexity of a many-body quantum dynamics for some particular cases.

\section{System Model and Problem Statement} 
\label{sec2}
\subsection{System Model}
Let $QG$ be the quantum gate structure of a gate-model quantum computer, defined with $L$ unitary gates, where an $i$-th, $i=1,\ldots ,L$ unitary gate $U_{i} \left(\theta _{i} \right)$ is 
\begin{equation} \label{1)} 
U_{i} \left(\theta _{i} \right)=\exp \left(-i\theta _{i} P_{i}\right),                                                               
\end{equation} 
where $P_{i}$ is a generalized Pauli operator formulated by the tensor product of Pauli operators $\left\{X,Y,Z\right\}$, while $\theta _{i} $ is the gate parameter associated with $U_{i} \left(\theta _{i} \right)$. 

The $L$ unitary gates formulate a system state ${| \vec{\theta } \rangle} $ of the quantum computer, as
\begin{equation} \label{ZEqnNum335090} 
{| \vec{\theta } \rangle} =U_{L} \left(\theta _{L} \right)U_{L-1} \left(\theta _{L-1} \right)\ldots U_{1} \left(\theta _{1} \right),                                         
\end{equation} 
where $U_{i} \left(\theta _{i} \right)$ identifies an $i$-th unitary gate and $\vec{\theta }$ is the collection of the gate parameters of the unitaries, defined as
\begin{equation} \label{ZEqnNum837426} 
\vec{\theta }=\left(\theta _{1} ,\ldots ,\theta _{L} \right)^{T} .     
\end{equation} 

The system state in \eqref{ZEqnNum335090} identifies a $U(\vec{\theta })$ unitary resulted from the product of the $L$ unitary operations $U_{L} \left(\theta _{L} \right)U_{L-1} \left(\theta _{L-1} \right)\ldots U_{1} \left(\theta _{1} \right)$ of the quantum computer. For an input quantum system $\left| \varphi  \right\rangle $, the $\left| \psi  \right\rangle $ output quantum system of $QG$ is as 
\begin{equation} 
\begin{split}
\left| \psi  \right\rangle &={| \vec{\theta } \rangle}\left| \varphi  \right\rangle \\
&=U(\vec{\theta })\left| \varphi  \right\rangle \\
&= U_{L} \left(\theta _{L} \right)U_{L-1} \left(\theta _{L-1} \right)\ldots U_{1} \left(\theta _{1} \right)\left| \varphi  \right\rangle .
\end{split}
\end{equation} 

The $f(\vec{\theta })$ objective function subject to a maximization is defined as
\begin{equation} \label{ZEqnNum830871} 
f( {\vec{\theta }})=\langle  {\vec{\theta }} |C(z)|\vec{\theta }\rangle ,
\end{equation} 
where $C\left(z\right)$ identifies a classical objective function \cite{ref10} of a computational problem, while $z$ is a bitstring resulting from an $M$ measurement. 

The $C$ classical objective function represents the objective function of a computational problem $\mathcal{P}$ fed into the quantum computer. The $C$ objective function is a subject of maximization via the quantum computer. Objective function examples are the combinatorial optimization problems \cite{ref9}, and the objective functions of large-scale programming problems \cite{nadd7}, such as the graph coloring problem, molecular conformation problem, job-shop scheduling problem, manufacturing cell formation problem, and the vehicle routing problem \cite{nadd7}.

At a target value $f^{*} (\vec{\theta })$, 
\begin{equation} \label{6)} 
{{f}^{*}}( {\vec{\theta }})=f( {{{\vec{\theta }}}^{*}})=\langle  {{{\vec{\theta }}}^{*}} | {{C}^{*}}(z)|{{{\vec{\theta }}}^{*}}\rangle ,
\end{equation} 
the problems are therefore to find a $\vec{\theta }^{*} $ that reaches the target state ${| \vec{\theta }^{*} \rangle} $ of the quantum computer and to identify the optimal $C^{*} \left(z\right)$ computational pathway for ${| \vec{\theta }^{*} \rangle} $.

\begin{definition}
(Computational pathway). The connectivity of $C\left(z\right)$ defines a computational pathway as the sum of $C_{ij} \left(z\right)$ objective function values evaluated between quantum states $ij$ in the $QG$ structure: 
\begin{equation} \label{5)} 
C\left(z\right)=\sum _{ij\in QG}C_{ij} \left(z\right) .                                                             
\end{equation} 
The $C \left(z\right)$ computational pathway between quantum states $ij$ sets the connectivity of objective function in a given state ${| \vec{\theta } \rangle}$ of the quantum computer.
\end{definition}
\begin{definition}
(Optimal computational pathway). The $C^{*} \left(z\right)$ optimal computational pathway of the quantum computer is the computational pathway associated with the optimal (target) state ${| \vec{\theta }^{*} \rangle}$. The $C^{*} \left(z\right)$ computational pathway sets the connectivity of the objective function in the target state ${| \vec{\theta }^{*} \rangle}$ of the quantum computer.
\end{definition}
\begin{definition}
(Connectivity graph of the quantum hardware).
The ${\rm {\mathcal G}}=\left(V,S\right)$ connectivity graph refers to the fixed connectivity of the hardvare of the $QG$ quantum gate structure, where the $v\in V$ nodes are quantum systems, while the $s\in S$ edges are the connections between them. An edge $s_{i,j} $ with index pair $\left(i,j\right)$ identifies a physical connection between quantum systems $v_{i} $ and $v_{j} $. 
\end{definition}

\subsection{Problem Statement}
The problem statement is given in Problems 1 and 2, as follows. 
\begin{problem}
(Target state determination of the quantum computer). For a given target objective function value $f(\vec{\theta }^{*})$, find the ${| \vec{\theta }^{*} \rangle} $ target state of the quantum computer from an initial state ${| \vec{\theta }_{0} \rangle} $ and an initial objective function $f(\vec{\theta }_{0})$. 
\end{problem}
\begin{problem}
(Computational path of the quantum computer in the target state). Determine the connectivity of the objective function $C^{*} \left(z\right)$ of $f(\vec{\theta }^{*})$ for the target quantum state ${| \vec{\theta }^{*} \rangle} $ of the quantum computer.
\end{problem}

Our solutions for Problems 1 and 2 are proposed in Theorems 1, 2, and Lemma 1.

\section{Results}
\label{sec3}
\subsection{Evaluation of the Target State of the Quantum Computer}
\begin{theorem}
(Target system state evaulation). The ${| \vec{\theta }^{*} \rangle} $ system state associated with the $f(\vec{\theta }^{*})$ target objective function can be evaluated from an initial state ${| \vec{\theta }_{0} \rangle} $ via a decomposition of the initial objective function $f(\vec{\theta }_{0})$.
\end{theorem}
\begin{proof}
Let $f(\vec{\theta }_{0})$ be the initial objective function value associated with ${| \vec{\theta }_{0} \rangle} $ and with gate parameters $\vec{\theta }_{0} $. The $f(\vec{\theta }_{0})$ value can be rewritten as
\begin{equation} \label{ZEqnNum989979} 
f(\vec{\theta }_{0})=(\vec{\theta }_{0})^{T} \chi ,                                                                
\end{equation} 
where $\chi $ is a vector of regression coefficients being evaluated via a ${\rm {\mathcal K}}$ kernel machine (see \eqref{a1)}), while $\vec{\theta }_{0} $ is decomposed as
\begin{equation} \label{ZEqnNum716686} 
\vec{\theta }_{0} =F(\vec{\theta }_{0})+F\left(U\right),                                                               
\end{equation} 
where $F(\vec{\theta }_{0})$ and $F\left(U\right)$ are orthogonal components, such that $F(\vec{\theta }_{0})$ depends on the actual objective function value, while $F\left(U\right)$ is a component independent from the current value of the objective function (i.e., $F\left(U\right)$ is a fixed component for an arbitrary $\vec{\theta }$) that lies in the null space. Since $\vec{\theta }_{0} $ and $f(\vec{\theta }_{0})$ are known, the $\chi $ regression coefficient vector can be determined from \eqref{ZEqnNum989979}.

Using \eqref{ZEqnNum716686}, the initial objective function in \eqref{ZEqnNum989979} can be rewritten at a particular $\chi $ as
\begin{equation} \label{ZEqnNum902493} 
f(\vec{\theta }_{0})=(F(\vec{\theta }_{0})+F(U))^{T} \chi ,                                                  
\end{equation} 
where the $F(\vec{\theta }_{0})$ component is evaluated at a given $\chi $ as
\begin{equation} \label{ZEqnNum214099} 
F(\vec{\theta }_{0})=\chi ^{+} f(\vec{\theta }_{0}),                                                              
\end{equation} 
where $+$ is the Moore--Penrose pseudoinverse \cite{ref25,ref35}. Since $F\left(U\right)$ has no dependence on the actual system state, it can be expressed from \eqref{ZEqnNum716686} and \eqref{ZEqnNum214099} as
\begin{equation} \label{ZEqnNum397382} 
F\left(U\right)=\vec{\theta }_{0} -F(\vec{\theta }_{0}).                                                           
\end{equation} 
Then, let $\vec{\theta }^{*} $ be the parameter vector associated with the target state ${| \vec{\theta }^{*} \rangle} $ of the target objective function $f(\vec{\theta }^{*})$.

Applying the same decomposition steps for the target $f(\vec{\theta }^{*})$, the component $F(\vec{\theta }^{*})$ at a given $\chi $ is 
\begin{equation} \label{ZEqnNum980531} 
F(\vec{\theta }^{*})=\chi ^{+} f(\vec{\theta }^{*}).                                                            
\end{equation} 
Therefore, the target vector $\vec{\theta }^{*} $ can be rewritten via \eqref{ZEqnNum980531} and \eqref{ZEqnNum397382} as
\begin{equation} \label{ZEqnNum175619} 
\vec{\theta }^{*} =F(\vec{\theta }^{*})+F(U)=\vec{\theta }_{0} +(\chi ^{+} f(\vec{\theta }^{*})-\chi ^{+} f(\vec{\theta }_{0})).                          
\end{equation} 
Using the $\vec{\theta }^{*} $ gate parameters in \eqref{ZEqnNum175619}, the target system state ${| \vec{\theta }^{*} \rangle} $ can be built up to achieve the target objective function $f(\vec{\theta }^{*})$. The target system state ${| \vec{\theta }^{*} \rangle} $ of a given $f(\vec{\theta }^{*})$ is therefore evolvable from the initial values $\vec{\theta }_{0} $,  $f(\vec{\theta }_{0})$, and $\chi $ that can be computed from \eqref{ZEqnNum989979}.  

Algorithm 1 summarizes the steps of the target system state evolution method.

\setcounter{algocf}{0}
\begin{algo}
  \DontPrintSemicolon
\caption{\textit{System state evolution of the quantum computer for a target objective function.}}

\textbf{Step 1}. Let $\vec{\theta }_{0} $ be the initial gate parameter vector and $f(\vec{\theta }_{0})$ be the initial objective function value.

\textbf{Step 2}. Set the target objective function value $f(\vec{\theta }^{*})$.

\textbf{Step 3}. Determine $\chi $ via the initial $f(\vec{\theta }_{0})$, as given by equation \eqref{ZEqnNum989979}. 

\textbf{Step 4}. Using $\chi $, determine $\vec{\theta }^{*} $ via \eqref{ZEqnNum175619} from $\vec{\theta }_{0} $, $f(\vec{\theta }_{0})$, and $f(\vec{\theta }^{*})$.

\textbf{Step 5}. Prepare the target state ${| \vec{\theta }^{*} \rangle} $ by the quantum computer. 

\end{algo}

\end{proof}

The results on the determination of the connectivity of the objective function in the target state are included in Theorem 2.

\subsection{Connectivity of the Objective Function in the Target State}
\begin{theorem}
(Connectivity of the objective function in the target state). The $\left(i,j\right)$ pairs of the $s_{i,j} $ edges of ${\rm {\mathcal G}}$, $\forall s_{i,j} \in S$, in a target objective function $C^{*} \left(z\right)=\sum _{\forall s_{i,j} \in S}C_{s_{i,j} }^{*} \left(z\right) $ associated to $f^{*} (\vec{\theta })$ can be determined from $\vec{\theta }^{*} $, where $C_{s_{i,j} }^{*} \left(z\right)$ is an objective function component associated to $s_{i,j} $.
\end{theorem} 
\begin{proof}
Let ${\rm {\mathcal G}}=\left(V,S\right)$ be the connectivity graph \cite{ref10} associated with the $QG$ quantum gate structure of the quantum computer (see Definition 3), and let $\vec{\theta }^{*} $ be evaluated as given in \eqref{ZEqnNum175619}. 
Let ${\rm {\mathcal X}}$ be the input space and let ${\rm {\mathcal K}}$ be a kernel machine, defined for a given $x,y\in {\rm {\mathcal X}}$ via kernel function \cite{uj1}, as
\begin{equation} \label{14)} 
{\rm {\mathcal K}}\left(x,y\right)=\Gamma \left(x\right)^{T} \Gamma \left(y\right),                                                           
\end{equation} 
where
\begin{equation} \label{ZEqnNum942922} 
\Gamma :{\rm {\mathcal X}}\to {\rm {\mathcal H}} 
\end{equation} 
is a nonlinear map from ${\rm {\mathcal X}}$ to the high-dimensional Reproducing Kernel Hilbert Space (RKHS) ${\rm {\mathcal H}}$ associated with ${\rm {\mathcal K}}$. Without loss of generality,
\begin{equation} \label{16)} 
\dim \left({\rm {\mathcal H}}\right){\rm \gg }\dim \left({\rm {\mathcal X}}\right),                                                         
\end{equation} 
and we assume that the map $\Gamma $ in \eqref{ZEqnNum942922} has no inverse. 

The connectivity of the objective function and the pairwise connectivity of the quantum computer's hardware are not related, since these connections are represented in different layers \cite{ref10}. While the physical-layer connectivity is determined by the $QG$ quantum gate structure of the fixed quantum hardware, the connectivity of the $C\left(z\right)$ objective function is determined in the logical-layer that formulates a computational pathway. As a corollary, the proposed algorithm works on fixed quantum hardware and iterates in the logical layer to determine the connectivity of the objective function such that the objective function is maximized.

Let $\vec{\kappa }$ be the vector of $s_{i,j} $ edges, $\forall s_{i,j} \in S$, and let $\vec{\Omega }$ be the vector of the actual $C_{s_{i,j} } \left(z\right)$ objective function values associated with the $s_{i,j} $ edges. The initial computational path of the quantum computer is therefore 
\begin{equation} \label{17)} 
C\left(z\right)=\sum _{\kappa _{i} }\Omega _{\kappa _{i} }  =\sum _{\forall s_{i,j} \in S}C_{s_{i,j} } \left(z\right) ,                                              
\end{equation} 
where $\kappa _{i} $ and $\Omega _{\kappa _{i} } $ identify the $i$-th elements of $\vec{\kappa }$ and $\vec{\Omega }$, respectively.

Then, let $\Upsilon _{0} $ be an element of the input space ${\rm {\mathcal X}}$, defined as 
\begin{equation} \label{18)} 
\Upsilon _{0} =(\vec{\kappa },\vec{\Omega })^{T} ,                                                             
\end{equation} 
and let $\tau _{0} $ be the map of $\Upsilon _{0} $ in ${\rm {\mathcal H}}$, as
\begin{equation} \label{ZEqnNum679565} 
\tau _{0} =\Gamma \left(\Upsilon _{0} \right)=\lambda \vec{\theta }_{0} ,                                                         
\end{equation} 
where $\lambda $ is a matrix of eigenvectors associated with the edge and objective function values in ${| \vec{\theta }_{0} \rangle} $.

Then, let $\Upsilon ^{*} $ be the target element in ${\rm {\mathcal X}}$ subject to be determined,
\begin{equation} \label{ZEqnNum216438} 
\Upsilon ^{*} =(\vec{\kappa }^{*} ,\vec{\Omega }^{*})^{T},
\end{equation} 
where $\vec{\kappa }^{*} $ and $\vec{\Omega }^{*} $ are target vectors that identify the connectivity of the $C_{s_{i,j} }^{*} \left(z\right)$ objective function values in the target state ${| \vec{\theta }^{*} \rangle} $, such that the $C^{*} \left(z\right)$ computational path can be evaluated as
\begin{equation} \label{ZEqnNum456970} 
C^{*} \left(z\right)=\sum _{\kappa _{i}^{*} }\Omega _{\kappa _{i}^{*} }^{*}  =\sum _{\forall s_{i,j} \in S}C_{s_{i,j} }^{*} \left(z\right) ,                                               
\end{equation} 
where $\kappa _{i}^{*} $ and $\Omega _{\kappa _{i}^{*} }^{*} $ refer to the $i$-th elements of  $\vec{\kappa }^{*} $ and $\vec{\Omega }^{*} $, respectively.

Then, let $\tau ^{*} $ be the map of the target $\Upsilon ^{*} \in {\rm {\mathcal X}}$ in ${\rm {\mathcal H}}$, defined as
\begin{equation} \label{ZEqnNum460233} 
\tau ^{*} =\Gamma \left(\Upsilon ^{*} \right)=\lambda ^{*} \vec{\theta }^{*} ,                                                       
\end{equation} 
where $\lambda ^{*} $ is a matrix of eigenvectors associated with the edge and objective function values in state ${| \vec{\theta }^{*} \rangle} $.

Since \eqref{ZEqnNum460233} is linear, in the ${| \vec{\theta }^{*} \rangle} $ state, the maps $\Gamma \left(\vec{\kappa }\right)$ and $\Gamma (\vec{\Omega })$ of $\vec{\kappa }^{*} $ and $\vec{\Omega }^{*} $, can be rewritten as
\begin{equation} \label{23)} 
\Gamma \left(\vec{\kappa }\right)=\mu \vec{\theta }^{*}  
\end{equation} 
and
\begin{equation} \label{24)} 
\Gamma (\vec{\Omega })=\nu \vec{\theta }^{*}  
\end{equation} 
with
\begin{equation} \label{25)} 
\lambda ^{*} =\left(\mu ,\nu \right)^{T} .                                                              
\end{equation} 
Since \eqref{ZEqnNum460233} can be evaluated from \eqref{ZEqnNum679565} in ${\rm {\mathcal H}}$, the task here is therefore to identify $\Upsilon ^{*} $ in ${\rm {\mathcal X}}$ from $\tau ^{*} $. As $\Upsilon ^{*} $ is determined, the target vectors $\vec{\kappa }^{*} $ and $\vec{\Omega }^{*} $ for the target objective function in \eqref{ZEqnNum456970} are also found.

Since the map $\Gamma $ in \eqref{ZEqnNum942922} has no inverse, finding $\Upsilon ^{*} $ in ${\rm {\mathcal X}}$ from $\tau ^{*} $ defines an ill-posed problem \cite{ref26,ref27,ref32,ref33,ref34}. In this setting, the determination of $\Upsilon ^{*} $ from $\tau ^{*} $, requires the use of a ${\rm {\mathcal P}}$ projector on $\tau _{0} $ \eqref{ZEqnNum679565} in ${\rm {\mathcal H}}$, which yields a ${\rm {\mathcal P}}\left(\tau _{0} \right)$ element in ${\rm {\mathcal H}}$. If $\tau ^{*} $ lies in (or close to) the span of $\left\{\Gamma \left(\Upsilon _{i} \right)\right\}$, where $\Upsilon _{i} $ is an $i$-th training data, $\Upsilon _{i}  \in {\rm {\mathcal X}}$, from a training set ${\rm {\mathcal S}}_{{\rm {\mathcal X}}}$ of $N$ training data, 
\begin{equation}
{\rm {\mathcal S}}_{{\rm {\mathcal X}}} =\left\{\Upsilon _{1} ,\ldots ,\Upsilon _{N} \right\}, 
\end{equation}
then $\tau ^{*} $ can be represented as a linear combination of the training data \cite{ref26,ref27,ref28}. As a corollary, ${\rm {\mathcal P}}\left(\tau _{0} \right)$ yields a close approximation of $\tau ^{*} $ in ${\rm {\mathcal H}}$:
\begin{equation} \label{26)} 
\tau ^{*} \approx {\rm {\mathcal P}}\left(\tau _{0} \right).                                                                
\end{equation} 
The ${\rm {\mathcal P}}\left(\tau _{0} \right)$ projection is defined as
\begin{equation} \label{ZEqnNum820370} 
{\rm {\mathcal P}}\left(\tau _{0} \right)=\sum _{i=1}^{n}\beta _{i}  V_{i} , 
\end{equation} 
where $V_{i} $ is a matrix of normalized eigenvectors of ${\rm {\mathcal K}}$, while $\beta _{i} $-s are projections as
\begin{equation} \label{28)} 
\beta _{i} =\sum _{j=1}^{N}\alpha _{j}^{i}  {\rm {\mathcal K}}\left(\Upsilon ^{*} ,\Upsilon _{j} \right),                                                        
\end{equation} 
while $\alpha _{i} $ is an $i$-th coefficient in the eigenvector $V$ as 
\begin{equation} \label{29)} 
V=\sum _{i=1}^{N}\alpha _{i} \tau _{i}  ,                                                               
\end{equation} 
where $\tau _{i} $ is the map of training data $\Upsilon _{i} $, as
\begin{equation} \label{30)} 
\tau _{i} =\Gamma \left(\Upsilon _{i} \right).                                                                
\end{equation} 
Then, based on \eqref{28)} and \eqref{29)}, a $j$-th component of $\chi $ from \eqref{ZEqnNum989979}, $\chi =\left\{ {{\chi }_{j}} \right\}_{j=1}^{N} $, can be determined as
\begin{equation} \label{a1)} 
{\chi }_{j}=\sum _{i=1}^{N}\tilde{\alpha }_{i}^{j} {\rm {\mathcal K}}\left(\Upsilon ^{*} ,\tilde{\Upsilon }_{i} \right) ,  
\end{equation} 
where $\tilde{\Upsilon }_{i} $ is a training data from a training set $\tilde{{\rm {\mathcal S}}}_{{\rm {\mathcal X}}} $, such that the constraint \cite{ref25,ref26} of 
\begin{equation} \label{aZEqnNum463012} 
\mu \left(\Gamma \left(\tilde{{\rm {\mathcal S}}}_{{\rm {\mathcal X}}} \right)\right)={\textstyle\frac{1}{N}} \sum _{j=1}^{N}\Gamma \left(\tilde{\Upsilon }_{j} \right)=0  
\end{equation} 
holds for $\tilde{{\rm {\mathcal S}}}_{{\rm {\mathcal X}}} $, where $\mu \left(\Gamma \left(\tilde{{\rm {\mathcal S}}}_{{\rm {\mathcal X}}} \right)\right)$ is the mean of the $\Gamma $-mapped training points $\tilde{{\rm {\mathcal S}}}_{{\rm {\mathcal X}}} $, while $\tilde{\alpha }_{i}^{j} $ is an $i$-th coefficient of a $j$-th eigenvector $\tilde{V}_{j}$,
\begin{equation} \label{a3)} 
\tilde{V}_{j}=\sum _{i=1}^{N}\tilde{\alpha }_{i}^{j} \Gamma \left(\tilde{\Upsilon }_{i} \right) .  
\end{equation} 
As it can be proven \cite{ref25,ref26,ref27}, the constraint in \eqref{aZEqnNum463012} satisfied, if the relation of
\begin{equation} \label{aZEqnNum247044} 
\left\langle \vec{K}\right\rangle \vec{\alpha }=N\lambda \vec{\alpha },  
\end{equation} 
holds for a particular training set ${\rm {\mathcal S}}_{{\rm {\mathcal X}}} $, where $\vec{\alpha }$ is the set of eigenvectors of $\vec{K}$ with eigenvalues $\lambda   $, while $\left\langle \vec{K}\right\rangle $ is the centered kernel matrix of ${\rm {\mathcal K}}$, defined as
\begin{equation} \label{a5)} 
\left\langle \vec{K}\right\rangle =\vec{K}-{\rm {\mathcal I}}\vec{K}-\vec{K}{\rm {\mathcal I}}+{\rm {\mathcal I}}\vec{K}{\rm {\mathcal I}}, 
\end{equation} 
where $\vec{K}$ is the kernel matrix of ${\rm {\mathcal K}}$, while ${\rm {\mathcal I}}$ is as
\begin{equation} \label{a6)} 
{\rm {\mathcal I}}=I-\vec{J}, 
\end{equation} 
where $I$ is the identity matrix, while $\vec{J}$ is an $N\times N$ matrix of ones.

Therefore, ${\chi }$ from \eqref{ZEqnNum989979} can be determined via the use of $\left\langle \vec{K}\right\rangle $ in \eqref{aZEqnNum247044} for a given ${\rm {\mathcal S}}_{{\rm {\mathcal X}}} $, which guarantees that \eqref{aZEqnNum463012} is satisfied, i.e., the $\Gamma \left({\rm {\mathcal S}}_{{\rm {\mathcal X}}} \right)$ mapped training data have zero mean that allows us to evaluate $\chi $ in an exact form.

The goal of projection ${\rm {\mathcal P}}$ is to minimize the $f_{d} \left(\tau ^{*} ,{\rm {\mathcal P}}\left(\tau _{0} \right)\right)$ distance in ${\rm {\mathcal H}}$, where
\begin{equation} \label{ZEqnNum492194} 
f_{d} \left(\tau ^{*} ,{\rm {\mathcal P}}\left(\tau _{0} \right)\right)=\left\| \tau ^{*} -{\rm {\mathcal P}}\left(\tau _{0} \right)\right\| ^{2} =\left\| \Gamma \left(\Upsilon ^{*} \right)-{\rm {\mathcal P}}\left(\tau _{0} \right)\right\| ^{2} .                       
\end{equation} 
Thus, at a given \eqref{ZEqnNum820370} and \eqref{ZEqnNum492194}, the term in \eqref{ZEqnNum216438} can be rewritten as an optimality criteria
\begin{equation} \label{32)} 
\Upsilon ^{*} =\mathop{\arg \min }\limits_{\Upsilon ^{*} \in {\rm {\mathcal X}}} f_{d} \left(\tau ^{*} ,{\rm {\mathcal P}}\left(\tau _{0} \right)\right).                                                
\end{equation} 
By introducing a non-negative regularization parameter $\Phi $ \cite{ref26} to weight the distance of ${{\left\| {{\Upsilon }^{*}}-{{\Upsilon }_{0}} \right\|}^{2}}$, the result in \eqref{ZEqnNum492194} at a given $\Upsilon _{0} \in {\rm {\mathcal X}}$ can be rewritten as
\begin{equation} \label{ZEqnNum730621} 
\begin{split}
   {{f}_{d}}\left( {{\tau }^{*}},\mathcal{P}\left( {{\tau }_{0}} \right) \right)\\&={{\left\| \Gamma \left( {{\Upsilon }^{*}} \right)-\mathcal{P}\left( {{\tau }_{0}} \right) \right\|}^{2}}+\Phi {{\left\| {{\Upsilon }^{*}}-{{\Upsilon }_{0}} \right\|}^{2}} \\ 
 & =\mathcal{K}\left( {{\Upsilon }^{*}},{{\Upsilon }^{*}} \right)-2\sum\limits_{i=1}^{N}{{{\ell }_{i}}}\mathcal{K}\left( {{\Upsilon }^{*}},{{\Upsilon }_{i}} \right) \\ 
 & +\Phi \left( {{\left( {{\Upsilon }^{*}} \right)}^{T}}{{\Upsilon }^{*}}+{{\left( {{\Upsilon }_{0}} \right)}^{T}}{{\Upsilon }_{0}}-2{{\Upsilon }^{*}}{{\Upsilon }_{0}} \right)+\zeta ,  
\end{split}
\end{equation} 
where $\zeta $ refers to terms independent of $\Upsilon ^{*} $, while $\ell _{i} $ is defined as
\begin{equation} \label{34)} 
\ell _{i} =\sum _{k=1}^{n}\beta _{k}  \alpha _{i}^{k} ,                                                                  
\end{equation} 
where $n$ is associated to the projection ${\rm {\mathcal P}}\left(\tau _{0} \right)$, since $\tau _{0} $ is projected to the subspace spanned by the first $n$ eigenvectors $V_{1} ,\ldots ,V_{q} $.

The result in \eqref{ZEqnNum730621} can be simplified by removing all terms independent of $\Upsilon ^{*} $, such that $f_{d} \left(\tau ^{*} ,{\rm {\mathcal P}}\left(\tau _{0} \right)\right)$ can be minimized for arbitrary ${\rm {\mathcal K}}$, as 
\begin{equation} \label{ZEqnNum927192} 
\begin{split}
   {{f}_{d}}\left( {{\tau }^{*}},\mathcal{P}\left( {{\tau }_{0}} \right) \right)=&\mathcal{K}\left( {{\Upsilon }^{*}},{{\Upsilon }^{*}} \right) \\ 
 & -2\sum\limits_{i=1}^{N}{{{\ell }_{i}}}\mathcal{K}\left( {{\Upsilon }^{*}},{{\Upsilon }_{i}} \right)+\Phi \left( {{\left( {{\Upsilon }^{*}} \right)}^{T}}{{\Upsilon }^{*}}-2{{\Upsilon }^{*}}{{\Upsilon }_{0}} \right),  
\end{split}
\end{equation} 
where 
\begin{equation} \label{36)} 
{\rm {\mathcal K}}\left(\Upsilon ^{*} ,\Upsilon ^{*} \right)=\Gamma \left(\Upsilon ^{*} \right)^{T} \Gamma \left(\Upsilon ^{*} \right)=\left(\tau ^{*} \right)^{T} \tau ^{*} .                                    
\end{equation} 
At a ${\rm {\mathcal P}}\left(\tau _{0} \right)$ with relation \eqref{ZEqnNum927192}, $\Upsilon ^{*} $ is determined as follows. Using \eqref{ZEqnNum927192} with an arbitrary ${\rm {\mathcal K}}$, $\Upsilon ^{*} $ can be evaluated as
\begin{equation} \label{ZEqnNum723849} 
\Upsilon ^{*} ={\textstyle\frac{1}{\tau ^{*} {\rm {\mathcal P}}\left(\tau _{0} \right)+\Phi }} \sum _{i=1}^{N}\ell _{i} {\rm {\mathcal K}}\left(\Upsilon ^{*} ,\Upsilon _{i} \right)\Upsilon _{i} +\Phi  \Upsilon _{0} , 
\end{equation} 
where the $\Phi $ regularization coefficient achieves the stability of $\Upsilon ^{*} $, while
\begin{equation} \label{38)} 
\tau ^{*} {\rm {\mathcal P}}\left(\tau _{0} \right)=\Gamma \left(\Upsilon ^{*} \right){\rm {\mathcal P}}\left(\Gamma \left(\Upsilon _{0} \right)\right)=\sum _{i=1}^{N}\ell _{i} {\rm {\mathcal K}}\left(\Upsilon ^{*} ,\Upsilon _{i} \right) ,                               
\end{equation} 
where ${\rm {\mathcal P}}\left(\tau _{0} \right)$ is defined in \eqref{ZEqnNum820370}. 

Then let ${\rm {\mathcal K}}'$ be the derivative of  ${\rm {\mathcal K}}$ such that it formulates the gradient with respect to $\Upsilon ^{*} $ as
\begin{equation} \label{ZEqnNum328288} 
\begin{split}
  & {{\nabla }_{{{\Upsilon }^{*}}}}\left( {{f}_{d}}\left( {{\tau }^{*}},\mathcal{P}\left( {{\tau }_{0}} \right) \right) \right) \\ 
 & =\sum\limits_{i=1}^{N}{{{\ell }_{i}}{\mathcal{K}}'\left( {{\Upsilon }^{*}},{{\Upsilon }_{i}} \right)}\left( {{\Upsilon }^{*}}-{{\Upsilon }_{i}} \right)+\Phi \left( {{\Upsilon }^{*}}-{{\Upsilon }_{0}} \right).  
\end{split}
\end{equation} 
As follows, for a $\vec{\theta }^{*} $, the target $\vec{\kappa }^{*} $ and $\vec{\Omega }^{*} $ can be determined for an arbitrary ${\rm {\mathcal K}}$ via a stable solution $\Upsilon ^{*} $ \eqref{ZEqnNum723849}, such that $\vec{\kappa }^{*} $ contains the $\left(i,j\right)$ pairs of the $s_{i,j} $ edges for $C_{s_{i,j} }^{*} \left(z\right)$, while $\vec{\Omega }^{*} $ identifies the values of $C_{s_{i,j} }^{*} \left(z\right)$ in ${\left| \theta ^{*}  \right\rangle} $. 

The proof is concluded here.
\end{proof}

\subsection{Computational Pathway of the Optimal State of the Quantum Computer}
\begin{lemma}
The $C^{*} \left(z\right)$ computational pathway of the optimal quantum state ${| \vec{\theta }^{*} \rangle} $ can be determined for an arbitrary ${\rm {\mathcal K}}$. 
\end{lemma}
\begin{proof}
To construct an iteration method for the determination of ${| \vec{\theta }^{*} \rangle} $ via $\Upsilon ^{*} $, some preliminary conditions are set as follows. For the ${\rm {\mathcal P}}\left(\tau _{0} \right)$ projection, we set the condition 
\begin{equation} \label{40)} 
{\rm {\mathcal P}}\left(\tau _{0} \right)\ne \vec{0},                                                                   
\end{equation} 
therefore
\begin{equation} \label{ZEqnNum689683} 
\tau ^{*} {\rm {\mathcal P}}\left(\tau _{0} \right)>0.                                                                 
\end{equation} 
Then, let $\varepsilon \left(\Upsilon ^{*} \right)$ be the extremum of $\Upsilon ^{*} $ defined \cite{ref27,ref28} as
\begin{equation} \label{42)} 
\varepsilon \left(\Upsilon ^{*} \right)={\textstyle\frac{1}{\sum _{j}\sigma _{j}  }} \sum _{i}\Upsilon _{i} \sigma _{i}  ,                                                         
\end{equation} 
where
\begin{equation} \label{43)} 
\sigma _{i} =\ell _{i} {\rm {\mathcal K}}'\left(\varepsilon \left(\Upsilon ^{*} \right),\Upsilon _{i} \right).                                                       
\end{equation} 
The gradient with respect to $\varepsilon \left(\Upsilon ^{*} \right)$ is 
\begin{equation} \label{44)} 
\nabla _{\varepsilon \left(\Upsilon ^{*} \right)} \left(f_{d} \left(\Gamma \left(\varepsilon \left(\Upsilon ^{*} \right)\right),{\rm {\mathcal P}}\left(\tau _{0} \right)\right)\right)=0.                                         
\end{equation} 
As ${\rm {\mathcal K}}$ is smooth, it can be shown that the condition of \eqref{ZEqnNum689683} always holds, since there is a neighborhood of the extremum \cite{ref26,ref27} of $f_{d} \left(\Gamma \left(\varepsilon \left(\Upsilon ^{*} \right)\right),{\rm {\mathcal P}}\left(\tau _{0} \right)\right)$. 

To provide the stability of $\Upsilon _{i}^{*} $ in an $i$-th iteration step, we utilize the $\Phi $ regularization coefficient from \eqref{ZEqnNum927192} for the evaluation $\Upsilon _{i}^{*} $, and for the computation the $f_{d}^{\left(i\right)} \left(\cdot \right)$ is the distance function associated to an $i$-th iteration step.

The steps are given in Algorithm 2.\\ 

 \setcounter{algocf}{1}
\begin{algo}
  \DontPrintSemicolon
\caption{\textit{Computational pathway of the optimal state of the quantum computer.}}

\textbf{Step 1}. Select an arbitrary ${\rm {\mathcal K}}$. Define training set ${\rm {\mathcal S}}_{{\rm {\mathcal X}}} =\left\{\Upsilon _{1} ,\ldots ,\Upsilon _{N} \right\}$ of $N$ elements of ${\rm {\mathcal X}}$, and their maps in ${\rm {\mathcal H}}$ as set ${\rm {\mathcal S}}_{{\rm {\mathcal H}}} =\left\{\Gamma \left(\Upsilon _{1} \right),\ldots ,\Gamma \left(\Upsilon _{N} \right)\right\}=\left\{\tau _{1} ,\ldots ,\tau _{N} \right\}$.

\textbf{Step 2}. Set the $R$ iteration number. For an $r$-th iteration step, $r=1,\ldots ,R$, evaluate $\Upsilon _{r}^{*} $ as   
\[\Upsilon _{r}^{*} ={\textstyle\frac{1}{\tau _{r-1}^{*} {\rm {\mathcal P}}\left(\tau _{0} \right)+\Phi }} \sum _{i=1}^{N}\ell _{i} {\rm {\mathcal K}}\left(\Upsilon _{r-1}^{*} ,\Upsilon _{i} \right)\Upsilon _{i} +\Phi  \Upsilon _{0} ,\] 
where 
\[\tau _{r-1}^{*} =\Gamma \left(\Upsilon _{r-1}^{*} \right),\] 
while $\Upsilon _{r-1}^{*} $ is the solution determined in the $\left(r-1\right)$-th iteration step, and
\[\tau _{r-1}^{*} {\rm {\mathcal P}}\left(\tau _{0} \right)=\Gamma \left(\Upsilon _{r-1}^{*} \right){\rm {\mathcal P}}\left(\Gamma \left(\Upsilon _{0} \right)\right)=\sum _{i=1}^{N}\ell _{i} {\rm {\mathcal K}}\left(\Upsilon _{r-1}^{*} ,\Upsilon _{i} \right) .\] 

\textbf{Step 3}. Compute $f_{d}^{\left(r\right)} \left(\tau _{r}^{*} ,{\rm {\mathcal P}}\left(\tau _{0} \right)\right)$ via \eqref{ZEqnNum927192} as
\[\begin{split}
   f_{d}^{\left( r \right)}\left( \tau _{r}^{*},\mathcal{P}\left( {{\tau }_{0}} \right) \right)=&\mathcal{K}\left( \Upsilon _{r}^{*},\Upsilon _{r}^{*} \right) \\ 
 & -2\sum\limits_{i=1}^{N}{{{\ell }_{i}}}\mathcal{K}\left( \Upsilon _{r}^{*},{{\Upsilon }_{i}} \right)\\&+\Phi \left( {{\left( \Upsilon _{r}^{*} \right)}^{T}}\Upsilon _{r}^{*}-2\Upsilon _{r}^{*}{{\Upsilon }_{0}} \right).  
\end{split}\]

\textbf{Step 4}. Repeat steps 2-3 for all $r$. 

\textbf{Step 5}. Output $\Upsilon ^{*} $ \eqref{ZEqnNum216438} as
\[\Upsilon ^{*} =\Upsilon _{R}^{*} ={\textstyle\frac{1}{\tau _{R-1}^{*} {\rm {\mathcal P}}\left(\tau _{0} \right)+\Phi }} \sum _{i=1}^{N}\ell _{i} {\rm {\mathcal K}}\left(\Upsilon _{R-1}^{*} ,\Upsilon _{i} \right)\Upsilon _{i} +\Phi  \Upsilon _{0} .\] 

\textbf{Step 6}. Set the connectivity of $C^{*} \left(z\right)$ in ${| \vec{\theta }^{*} \rangle} $ via $\Upsilon ^{*} $.

\end{algo}

\end{proof}

\section{Conclusions}
\label{sec4}
Gate-model quantum computers represent an implementable way for near-term experimental quantum computations. The resolution of a computational problem fed into a quantum computer can be modeled via reaching the target value of an objective function. The objective function is determined by the actual computational problem. To satisfy the target objective function value, a quantum computer must reach a target system state. In the target system state, the gate parameters of the unitaries pick up values that set the objective function into the target value. Finding the target system state is a challenge that requires several rounds of measurement and system state preparations via the quantum computer. Here, we proved that the target state of the quantum computer can be evaluated from an initial system state and an initial objective function. The solution significantly reduces the cost of objective function evaluation, since the proposed method requires no the preparation of intermediate system states via the quantum computer between the initial and target system states. We defined a method for the evaluation of the computational path of the quantum computer for the target state, and an algorithm to solve the computational path problem in an iterative manner.


\section*{Acknowledgements}
The research reported in this paper has been supported by the Hungarian Academy of Sciences (MTA Premium Postdoctoral Research Program 2019), by the National Research, Development and Innovation Fund (TUDFO/51757/2019-ITM, Thematic Excellence Program), by the National Research Development and Innovation Office of Hungary (Project No. 2017-1.2.1-NKP-2017-00001), by the Hungarian Scientific Research Fund - OTKA K-112125 and in part by the BME Artificial Intelligence FIKP grant of EMMI (Budapest University of Technology, BME FIKP-MI/SC).

\newpage
\onecolumn
\appendix
\setcounter{table}{0}
\setcounter{figure}{0}
\setcounter{equation}{0}
\setcounter{algocf}{0}
\renewcommand{\thetable}{\Alph{section}.\arabic{table}}
\renewcommand{\thefigure}{\Alph{section}.\arabic{figure}}
\renewcommand{\theequation}{\Alph{section}.\arabic{equation}}
\renewcommand{\thealgocf}{\Alph{section}.\arabic{algocf}}

\setlength{\arrayrulewidth}{0.1mm}
\setlength{\tabcolsep}{5pt}
\renewcommand{\arraystretch}{1.5}
\section{Appendix}

\subsection{Abbreviations}
\begin{description}
\item[NISQ] Noisy Intermediate-Scale Quantum
\item[QAOA] Quantum Approximate Optimization Algorithm
\item[RKHS] Reproducing Kernel Hilbert Space
\end{description}

\subsection{Notations}
The notations of the manuscript are summarized in \tref{tab2}.
\begin{center}
\begin{longtable}{||l|p{4.7in}||}
\caption{Summary of notations.}
\label{tab2}
\endfirsthead
\endhead
\hline
\textit{Notation} & \textit{Description} \\ \hline
$QG$ & Quantum gate structure of a gate-model quantum computer. \\ \hline 
$U_{i} \left(\theta _{i} \right)$ & An $i$-th unitary gate, $U_{i} \left(\theta _{i} \right)=\exp \left(-i\theta _{i} P_{i}\right)$, where $P_{i}$ is a generalized Pauli operator formulated by a tensor product of Pauli operators $\left\{X,Y,Z\right\}$, while $\theta _{i} $ is referred to as the gate parameter associated to $U_{i} \left(\theta _{i} \right)$. \\ \hline 
${| \vec{\theta } \rangle} $ & System state of the quantum computer, ${| \vec{\theta } \rangle} =U_{L} \left(\theta _{L} \right)U_{L-1} \left(\theta _{L-1} \right)\ldots U_{1} \left(\theta _{1} \right)$, where $U_{i} \left(\theta _{i} \right)$ identifies an $i$-th unitary gate. \\ \hline 
$\vec{\theta }$ & Gate parameter vector, a collection of gate parameters of the $L$ unitaries, $\vec{\theta }=\theta _{L} ,\theta _{L-1} ,\ldots ,\theta _{1} $. \\ \hline 
$C\left(z\right)$ & Objective function of a computational problem fed into the quantum computer. It identifies the computational pathway (connectivity of the objective function as) $C\left(z\right)=\sum _{ij\in QG}C_{ij} \left(z\right) $, where $C_{ij} \left(z\right)$ is evaluated between quantum states $ij$ in the $QG$ structure of the gate-model quantum computer. \\ \hline 
$C^{*} \left(z\right)$ & Computational pathway in the target state ${| \vec{\theta }^{*} \rangle} $. \\ \hline 
$z$ & A bitstring. \\ \hline 
$f(\vec{\theta })$ & Objective function. \\ \hline 
$f^{*} (\vec{\theta })$ & A target objective function value. \\ \hline 
${| \vec{\theta }_{0} \rangle} $ & Initial system state of the quantum computer. \\ \hline 
${| \vec{\theta }^{*} \rangle} $ & Target system state of the quantum computer subject to be determined that achieves $f^{*} (\vec{\theta })$. \\ \hline 
$\chi $ & Vector of regression coefficients. \\ \hline 
$\vec{\theta }_{0} $ & Collection of gate parameters in the ${| \vec{\theta }_{0} \rangle} $ initial system state. \\ \hline 
$\vec{\theta }^{*} $ & Collection of gate parameters in the ${| \vec{\theta }^{*} \rangle} $ target system state. \\ \hline 
$F(\vec{\theta }_{0})$ & Component of $\vec{\theta }_{0} $. \\ \hline 
$F\left(U\right)$ & Fixed component for an arbitrary $\vec{\theta }$. \\ \hline 
$+$ & Moore--Penrose pseudoinverse. \\ \hline 
${\rm {\mathcal G}}$ & Connectivity graph, ${\rm {\mathcal G}}=\left(V,S\right)$, with a set $V$ of vertexes, and a set $S$ of arcs. \\ \hline 
$V$ & Set of vertexes in ${\rm {\mathcal G}}$. \\ \hline 
$S$ & Set of arcs in ${\rm {\mathcal G}}$. \\ \hline 
$v$ & A vertex of $V$ the ${\rm {\mathcal G}}$ environmental graph. \\ \hline 
$s_{i,j} $ & Edge $s_{i,j} $ with index pair $\left(i,j\right)$, it identifies a connection between nodes $v_{i} $ and $v_{j} $.  \\ \hline 
${\rm {\mathcal X}}$ & Input space. \\ \hline 
${\rm {\mathcal K}}$ & Kernel machine. \\ \hline 
${\rm {\mathcal H}}$ & Reproducing Kernel Hilbert Space (RKHS) associated with the kernel machine ${\rm {\mathcal K}}$. \\ \hline 
$\Gamma $ & A nonlinear map, $\Gamma :{\rm {\mathcal X}}\to {\rm {\mathcal H}}$, from ${\rm {\mathcal X}}$ to the high-dimensional Hilbert space ${\rm {\mathcal H}}$ associated with ${\rm {\mathcal K}}$. \\ \hline 
$\vec{\kappa }$ & Vector of initial $s_{i,j} $ edges in the ${\rm {\mathcal G}}$ connectivity graph. \\ \hline 
$\vec{\Omega }$ & Vector of the actual $C_{s_{i,j} } \left(z\right)$ objective function values, associated to the $s_{i,j} $ edges in the ${\rm {\mathcal G}}$ connectivity graph. \\ \hline 
$\kappa _{i} $ & An $i$-th element of $\vec{\kappa }$. \\ \hline 
$\Omega _{\kappa _{i} } $ & An $i$-th element of $\vec{\Omega }$. \\ \hline 
$\Upsilon _{0} $ & Initial element in the input space ${\rm {\mathcal X}}$, defined as $\Upsilon _{0} =(\vec{\kappa },\vec{\Omega })^{T} $. \\ \hline 
$\tau _{0} $ & Map of $\Upsilon _{0} $ in ${\rm {\mathcal H}}$, $\tau _{0} =\Gamma \left(\Upsilon _{0} \right)=\lambda \vec{\theta }_{0} $, where $\lambda $ is a matrix of eigenvectors associated with the edge and objective function values in ${| \vec{\theta }_{0} \rangle} $. \\ \hline 
$\Upsilon ^{*} $ & Target element in ${\rm {\mathcal X}}$ subject to be determined, $\Upsilon ^{*} =(\vec{\kappa }^{*} ,\vec{\Omega }^{*})^{T} $. \\ \hline 
$\vec{\kappa }^{*} $ & Target vector, identifies the connectivity of the $C_{s_{i,j} }^{*} \left(z\right)$ objective function values in the target state ${| \vec{\theta }^{*} \rangle} $. \\ \hline 
$\vec{\Omega }^{*} $ & Target vector, identifies the connectivity of the $C_{s_{i,j} }^{*} \left(z\right)$ objective function values in the target state ${| \vec{\theta }^{*} \rangle} $. \\ \hline 
$\kappa _{i}^{*} $ & An $i$-th element of  $\vec{\kappa }^{*} $ and $\vec{\Omega }^{*} $. \\ \hline 
$\Omega _{\kappa _{i}^{*} }^{*} $ & An $i$-th element of $\vec{\Omega }^{*} $. \\ \hline 
$\tau ^{*} $ & Map of the target $\Upsilon ^{*} \in {\rm {\mathcal X}}$, $\tau ^{*} =\Gamma \left(\Upsilon ^{*} \right)=\lambda ^{*} \vec{\theta }^{*} $, where $\lambda ^{*} $ is a matrix of eigenvectors associated with the edge and objective function values in state ${| \vec{\theta }^{*} \rangle} $. \\ \hline 
${\rm {\mathcal P}}$ & Projector in ${\rm {\mathcal H}}$.  \\ \hline 
$\Upsilon _{0} $ & Initial element in ${\rm {\mathcal X}}$. \\ \hline 
$\Upsilon _{i} $ & Training data in ${\rm {\mathcal X}}$. \\ \hline 
$\Upsilon ^{*} $ & Target element in ${\rm {\mathcal X}}$. \\ \hline 
$V $ & An eigenvector. \\ \hline 
$\beta _{i} $ & Projections in ${\rm {\mathcal H}}$, $\beta _{i} =\sum _{j=1}^{N}\alpha _{j}^{i} {\rm {\mathcal K}}\left(\Upsilon ^{*} ,\Upsilon _{j} \right) $, where $\alpha _{i} $ is an $i$-th coefficient in the eigenvector $V$, $V=\sum _{i=1}^{N}\alpha _{i} \tau _{i}  $, where $\tau _{i} $ is the map of training data $\Upsilon _{i} $, $\tau _{i} =\Gamma \left(\Upsilon _{i} \right)$. \\ \hline 
$f_{d} \left(x,y\right)$ & Distance function in ${\rm {\mathcal H}}$, $f_{d} \left(x,y\right)=\left\| x-y\right\| ^{2} $. \\ \hline 
$\Phi $ & A non-negative regularization parameter. \\ \hline 
$\zeta $ & Terms independent of $\Upsilon ^{*} $. \\ \hline 
$\ell _{i} $ & Parameter, $\ell _{i} =\sum _{k=1}^{n}\beta _{k} \alpha _{i}^{k}  $, where $n$ is associated to the projection ${\rm {\mathcal P}}\left(\tau _{0} \right)$. \\ \hline 
$\varepsilon \left(\Upsilon ^{*} \right)$ & Extremum of $\Upsilon ^{*} $, $\varepsilon \left(\Upsilon ^{*} \right)={\textstyle\frac{1}{\sum _{j}\sigma _{j}  }} \sum _{i}\Upsilon _{i} \sigma _{i}  $,  where $\sigma _{i} =\ell _{i} {\rm {\mathcal K}}'\left(\varepsilon \left(\Upsilon ^{*} \right),\Upsilon _{i} \right)$. \\ \hline 
$\nabla _{\varepsilon \left(\Upsilon ^{*} \right)} \left(f_{d} \left(\cdot \right)\right)$ & Gradient with respect to $\varepsilon \left(\Upsilon ^{*} \right)$. \\ \hline 
$f_{d}^{\left(i\right)} \left(\cdot \right)$ & Distance function associated to an $i$-th iteration step. \\ \hline 
${\rm {\mathcal S}}_{{\rm {\mathcal X}}} $ & Training set of $N$ training data in ${\rm {\mathcal X}}$, ${\rm {\mathcal S}}_{{\rm {\mathcal X}}} =\left\{\Upsilon _{1} ,\ldots ,\Upsilon _{N} \right\}$. \\ \hline 
${\rm {\mathcal S}}_{{\rm {\mathcal H}}} $ & Set of maps of the training data in ${\rm {\mathcal H}}$,  ${\rm {\mathcal S}}_{{\rm {\mathcal H}}} =\left\{\Gamma \left(\Upsilon _{1} \right),\ldots ,\Gamma \left(\Upsilon _{N} \right)\right\}=\left\{\tau _{1} ,\ldots ,\tau _{N} \right\}$. \\ \hline 
$R$ & Iteration number. \\ \hline 
$\Upsilon _{r}^{*} $ & Target value $\Upsilon ^{*} $ associated with an $r$-th iteration step, $r=1,\ldots ,R$. \\ \hline 
$\Upsilon _{R}^{*} $ & Solution determined in the $R$-th iteration step, $\tau _{R}^{*} =\Gamma \left(\Upsilon _{R}^{*} \right)$. \\ \hline
\end{longtable}
\end{center}
\end{document}